\documentclass[12pt,a4paper]{article}
\usepackage{standalone} 
\usepackage[a4paper,margin=1in]{geometry}
\usepackage{amsmath,amsfonts,amssymb}
\usepackage{amsthm}
\graphicspath{{fig/}}     
\usepackage{graphicx}
\usepackage{hyperref}
\usepackage{xcolor}
\usepackage{orcidlink}
\usepackage{comment}
\usepackage{float}
\usepackage{amsfonts} 
\newcommand{\R}{\mathbb{R}}
\usepackage[font=small,labelfont=bf]{caption}
\usepackage{natbib}

\newtheorem{theorem}{Theorem}
\newtheorem{lemma}[theorem]{Lemma}
\newtheorem{proposition}[theorem]{Proposition}
\usepackage{hyperref}
\usepackage{todonotes}
\usepackage{geometry}
\usepackage{doi}
\usepackage{natbib}
\usepackage{authblk}

\title{Deterministic and stochastic infection dynamics in a population subject to stress}
\author[1]{Clotilde Djuikem\orcidlink{0000-0003-4815-743X
}\thanks{Corresponding author: \href{mailto:Clotilde.Djuikem@umanitoba.ca}{Clotilde.Djuikem@umanitoba.ca}}}
\author[1]{Julien Arino\orcidlink{0000-0001-6409-5027}}
\affil[1]{ Department of Mathematics, University of Manitoba, Winnipeg, Manitoba, Canada}
\date{}

\begin{document}
\maketitle


\begin{abstract}
Physiological stress fundamentally alters disease susceptibility in aquatic environments. In this paper, we develop a stress-structured epidemiological model where host vulnerability is dynamically driven by water quality.
Analytically, we establish that the system exhibits a classic forward bifurcation at $\mathcal{R}_0=1$, confirming that the basic reproduction number remains a valid threshold for eradication.
However, stochastic analysis reveals a critical asymmetry not captured by deterministic thresholds. We show that while $\mathcal{R}_0$ predicts stability, the probability of an outbreak depends on the initial physiological state. Introducing infection into a stressed sub-population leads to immediate rapid growth of the disease, whereas introduction into the normal class faces a stochastic barrier that significantly delays the epidemic peak.
\end{abstract}

\section{Introduction}

Infectious diseases in aquatic systems result in substantial economic losses in aquaculture worldwide and impact biodiversity in the wild \citep{subasinghe2009global}. 
In the current context of climate change and intensifying production, water quality fluctuations are becoming more frequent \cite{reid2019climate}. 
Hence, interactions between environmental drivers, host physiology and pathogens need to be better understood to control outbreaks.

This work is motivated by the physiological link between environmental stress and host susceptibility.
Exposure to stressors, particularly hypoxia (low dissolved oxygen), induces a cortisol response in fish. 
This endocrine response suppresses the mucosal immunity of the skin and gills, which serves as the primary defense against waterborne pathogens \citep{snieszko1974stress,tort2011stress}. Consequently, a decline in water quality can rapidly transition a population from a resistant state to a highly susceptible one.

Mathematical models of aquatic diseases have been widely studied. For example, the spread of sea lice in salmon farms was modelled in~\cite{krkosek2005transmission}, showing that management strategies, such as mandatory production breaks, are essential for control. 
Theoretical models were built in~\cite{lafferty2003how} to test whether environmental stress increases or decreases disease prevalence, finding that the effect depends strongly on host density. 
It was emphasized in~\cite{harvell2002climate} that climate warming and unusual environmental events can activate new diseases in marine organisms, pointing to the role of abiotic stressors. 
Finally, how seasonal changes in the environment affect epidemic cycles was studied in~\cite{altizer2006seasonality}, showing that transmission rates vary over time with external drivers.

However, few models focus on the \emph{internal} physiological state of the host. 
Classical compartmental models, such as the Susceptible--Infected--Recovered (SIR) framework, usually treat host susceptibility as fixed \cite{kermack1927contribution}. 
Even heterogeneous models, such as those in~\cite{dwyer1997host}, which separate robust and frail individuals, assume these traits do not change; individuals remain in the same risk class for life.

An important question in modelling heterogeneous populations is whether such structure changes the disease-free state. Many studies highlight the risk of \emph{backward bifurcation}, where disease can persist even when the basic reproduction number $\mathcal{R}_0$ is below one \cite{hadeler1997backward}. However, the link between this theory and host physiology still needs attention to develop best control strategies. 
It was shown experimentally in~\cite{gervasi2015host} that stress hormones can change host competence, but most models include this only as a parameter change, not as a shift in population structure. 
As a result, current models rarely capture the \textit{dynamic transition} between health states. 
In particular, it remains unclear how the movement of hosts from a ``normal'' to a ``stressed'' state due to water quality affects infection severity and the timing of disease introduction. 

To bridge this gap, we propose a mathematical framework that explicitly couples stress dynamics with disease transmission.
We develop and analyse both the deterministic and stochastic dynamics of the system. 

The paper is organised as follows. 
In Section~\ref{sec:deterministic}, we introduce the deterministic stress–infection model, establish basic qualitative properties and derive the basic reproduction number and the stability of the equilibria. 
The second part of that section extends the analysis to time-dependent stress, introducing a time-varying reproduction number and bounds based on constant-stress regimes. 
In Section~\ref{sec:stochastic}, we formulate the stochastic continuous-Time Markov chain version of the model, analyse the associated multitype branching process, and explore extinction probabilities, outbreak sizes and first-introduction times under different water–stress scenarios. 
We conclude in Section~\ref{sec:discussion} with a discussion of the implications for stress management and disease control in fish populations and other host–pathogen systems.

\section{Deterministic model and analysis}\label{sec:deterministic} 

\subsection{Model formulation}

We consider a structured fish population composed of normal and stressed individuals, in which stress modifies both susceptibility and recovery. 
Let $S_N$ and $S_S$ denote, respectively, the densities of normal and stressed susceptible fish, $I_N$ and $I_S$ the corresponding infected classes and $R$ the recovered class.
The total population at time $t$ is
\[
N(t) = S_N(t) + S_S(t) + I_N(t) + I_S(t) + R(t).
\]
Births occur at a constant rate $\Lambda$; all individuals die naturally at \emph{per capita} rate $\mu$. The schematic representation of the deterministic system is given in Figure~\ref{fig:model-flowchart}, highlighting the stress transition ($\alpha(t)$), infection processes ($\beta_N, \beta_S$), and recovery ($\gamma_N, \gamma_S$) and mortality due to the disease ($d_N, d_S$). Since our primary interest is to understand the initial introduction of the pathogen within susceptible classes, we assume that there are no transitions between infected classes due to the fact that the impact of environmental stress on already infected individuals is captured within the parameter differences between the two classes. We also assume that stress acts irreversibly on the epidemic time scale considered here. While susceptible individuals may enter the stressed class in response to deteriorating environmental conditions, recovery from stress would require sustained environmental improvement and physiological recovery, which occur on longer time scales and are therefore neglected.

\begin{figure}[H]
    \centering
    \includegraphics{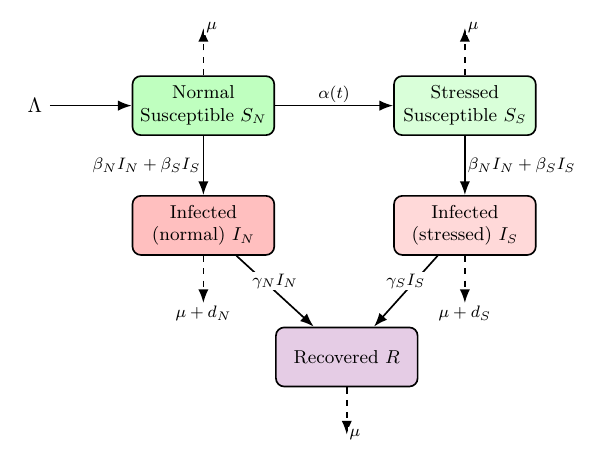} 
    \caption{Flow diagram of the deterministic stress–infection model:
    Normal susceptible ($S_N$), 
    Stressed susceptible ($S_S$), 
    Normal infected ($I_N$), 
    Stressed infected ($I_S$)
    and Recovered ($R$).}
    \label{fig:model-flowchart}
\end{figure}
The infection process is described by a mass-action force of infection
\[
\lambda(t) = \beta_N I_N(t) + \beta_S I_S(t),
\]
where $\beta_N$ and $\beta_S$ are transmission parameters for normal and stressed susceptible classes, respectively.
Disease-induced mortality occurs at \emph{per capita} rates $d_N$ and $d_S$ for normal and stressed infected compartments, respectively; recovery from infection in these compartments occurs at the \emph{per capita} rates $\gamma_N$ and $\gamma_S$, respectively.
The resulting stress--infection system is given  in~\eqref{eq:model-ODE-alpha-t}, with parameter  names and values in Table~\ref{tab:parameters}:
\begin{equation}\label{eq:model-ODE-alpha-t}
\begin{cases}
\displaystyle \frac{dS_N}{dt} = \Lambda - \alpha(t) S_N - \lambda(t) S_N - \mu S_N,\\[6pt]
\displaystyle \frac{dS_S}{dt} = \alpha(t) S_N - \lambda(t) S_S - \mu S_S,\\[6pt]
\displaystyle \frac{dI_N}{dt} = \lambda(t) S_N - (\gamma_N+\mu+d_N) I_N,\\[6pt]
\displaystyle \frac{dI_S}{dt} = \lambda(t) S_S - (\gamma_S+\mu+d_S) I_S,\\[6pt]
\displaystyle \frac{dR}{dt}   = \gamma_N I_N + \gamma_S I_S - \mu R,
\end{cases}
\end{equation}
with non-negative initial conditions $(S_N(0),S_S(0),I_N(0),I_S(0),R(0))\in\mathbb{R}_+^5$.

\subsubsection{Water conditions and stress}

To represent environmental conditions, we introduce a water-quality variable $W(t)$, interpreted here as dissolved oxygen (DO) concentration in the water column, measured in mg\,L$^{-1}$.
Fish populations typically exhibit non-linear physiological responses to hypoxia: they maintain homeostasis until oxygen levels drop below a specific threshold, after which stress responses escalate rapidly.
To capture this biological threshold effect, we define a dimensionless water-stress index $S_W(t) \in [0,1]$ using a sigmoidal function adapted from the general stress formulation in \cite{pfister2009assessing}:

\begin{equation}\label{eq:water-stress-index}
S_W(t) = \frac{1}{1 + \exp\big(-k_W\,(W_{\mathrm{crit}} - W(t))\big)}.
\end{equation}

Here, $W_{\mathrm{crit}}$ represents the critical dissolved oxygen level marking the onset of hypoxia-driven stress.
The parameter $k_W > 0$ determines the steepness of this transition; a high value of $k_W$ implies that a small decrease in oxygen near $W_{\mathrm{crit}}$ triggers a sharp shift from the normal state ($S_W \approx 0$) to the stressed state ($S_W \approx 1$).

We link this environmental driver to the epidemic model~\eqref{eq:model-ODE-alpha-t} by defining the transition rate from the normal to the stressed susceptible compartment as
\begin{equation}\label{eq:alpha-time-dependent}
\alpha(t) = \alpha_{\max}\,S_W(t),
\end{equation}
where $\alpha_{\max} > 0$ is the maximal stress induction rate.
In this formulation, water quality acts as an external function that modulates the flow of susceptible hosts into the high-risk class $S_S$.


To illustrate how this non-linear threshold transforms environmental data into physiological risk, we consider four distinct water-quality regimes.
Figure~\ref{fig:W-alpha-scenarios} shows four curves describing these four scenarios, ranging from well-oxygenated to persistently low-oxygen conditions, and the resulting stress rate. 
By comparing the dissolved oxygen profiles  (left panel) with the stress rates  (right panel), we observe that stress remains negligible in well-oxygenated conditions but rises sharply once dissolved oxygen falls near or below $W_{\text{crit}}$.

\begin{figure}[ht]
    \centering
    \includegraphics[width=\textwidth]{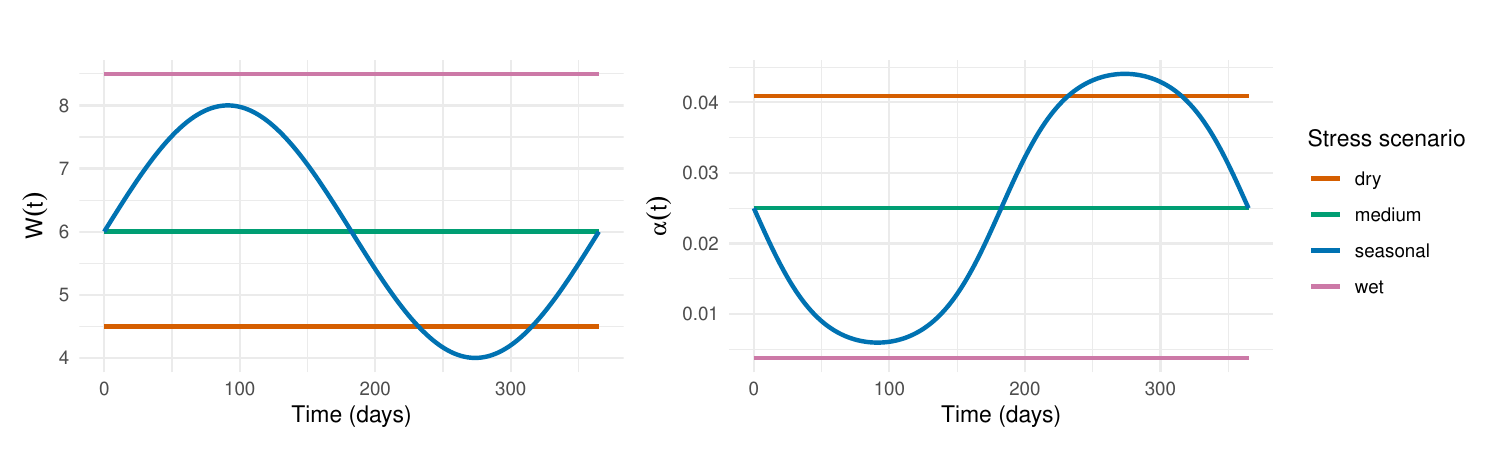}
    \caption{
    Dissolved oxygen profiles $W(t)$ (left panel) and corresponding stress rates $\alpha(t)$ (right panel) under four illustrative water-quality scenarios: a well-oxygenated ``wet'' regime (8.5 mg\,L$^{-1}$), a borderline ``medium'' regime at the critical threshold ($W_{\mathrm{crit}}=6$ mg\,L$^{-1}$), a low-oxygen ``dry'' regime (4.5 mg\,L$^{-1}$) \cite{pfister2009assessing} and a ``seasonal'' regime in which $W(t)$ oscillates between approximately 4 and 8 mg\,L$^{-1}$ over one year. The remain parameter values are given in Table~\ref{tab:parameters}. }
    \label{fig:W-alpha-scenarios}
\end{figure}

Experimental studies compiled in \citep{RandallMacKinlay2002} show that
fish responses to hypoxia span multiple time scales, from rapid physiological adjustments
(hours) to longer-lasting molecular and hematological changes (days to weeks). We therefore
interpret $\alpha_{\max}$ as the characteristic timescale of stress induction under sustained
hypoxia. Accordingly,
we set $\alpha_{\max}=0.05~\mathrm{day}^{-1}$, which corresponds to a characteristic induction
time of approximately $1/\alpha_{\max}\approx 20$ days.

\begin{table}[h!]
\centering
\caption{Description of model parameters.}
\begin{tabular}{llll}
\hline
\textbf{Symbol} & \textbf{Description} & \textbf{Value/Unit} & \textbf{Source}\\
\hline
$\Lambda$ & Recruitment or birth rate of fish & 10 ind$\cdot$day$^{-1}$ & Assumed\\
$\alpha_{\max}$ & Maximal stress rate  & 0.05 day$^{-1}$ & \citep{RandallMacKinlay2002}\\
$W_{\mathrm{crit}}$ & Critical dissolved oxygen level & 6 mg$\cdot$L$^{-1}$ & \citep{pfister2009assessing, solstorm2016dissolved} \\
$k_W$ & Sensitivity of stress to oxygen deficit & 1 (mg$\cdot$L$^{-1}$)$^{-1}$ & Assumed\\
$\beta_N$ & Transmission rate for normal  fish & ind$^{-1}$.day$^{-1}$ & $f(\mathcal R_{01})$\\
$\beta_S$ & Transmission rate for stressed fish & ind$^{-1}$.day$^{-1}$ & $2\beta_N$\\
$\gamma_N$ & Recovery rate for normally infected fish & 0.1 day$^{-1}$ & \citep{ogut2005dynamics}\\
$\gamma_S$ & Recovery rate for stressed infected fish & 0.1 day$^{-1}$ & \citep{ogut2005dynamics}\\
$\mu$ & Natural mortality rate & $10^{-4}-2.10^{-3}$ day$^{-1}$ & \citep{fjelldal2019effects} \\
$d_N$ & Disease-induced mortality for $I_N$ fish & $4.10^{-3}-3.10^{-2}$ day$^{-1}$ & \cite{fjelldal2019effects} \\
$d_S$ & Disease-induced mortality for $I_S$ fish &$4.10^{-3}-3.10^{-2}$ day$^{-1}$ & \cite{fjelldal2019effects}\\
\hline
\end{tabular}
\label{tab:parameters}
\end{table}

\subsubsection{Basic properties of the ODE model}

\begin{lemma}
Assume that $\alpha(t)$ is bounded and piecewise
continuous on $[0,\infty)$, and that all parameters are nonnegative.
Then, for any initial condition
\[
\mathbf X(0)=(S_N(0),S_S(0),I_N(0),I_S(0),R(0))\in\R_+^5,
\]
system~\eqref{eq:model-ODE-alpha-t} admits a unique solution defined for
all $t\ge 0$, and this solution remains nonnegative:
\[
\mathbf X(t)=(S_N(t),S_S(t),I_N(t),I_S(t),R(t))\ge 0,\qquad t\ge 0.
\]
\end{lemma}

\begin{proof}
The right-hand side of~\eqref{eq:model-ODE-alpha-t} is continuous in $t$
and locally Lipschitz in $\mathbf X(t)$, being a polynomial in
the state variables with bounded time-dependent coefficient
$\alpha(t)$. Hence existence and uniqueness of a maximal
solution follow from the Picard--Lindel\"of theorem
\cite{Perko2001}.
To prove positivity, suppose that one of the variables, that we will denote $X_i$
is equal to 0 at some
time, with all others being non-negative. A quick analysis shows that $\dot X_i\geq 0$ so that $X_i$ cannot become negative \cite{FarinaRinaldi2000}.

\end{proof}

\begin{lemma}
The domain
\[
\Omega = \left\{\mathbf X(t)\in\R_+^5:
N\le \dfrac{\Lambda}{\mu}\right\}
\]
is positively invariant under the flow of \eqref{eq:model-ODE-alpha-t} and attracts all trajectories starting in
$\R_+^5$.
\end{lemma}

\begin{proof}
Summing the equations of~\eqref{eq:model-ODE-alpha-t} yields
\[
\frac{dN}{dt}
= \Lambda - \mu N - d_N I_N - d_S I_S
\le \Lambda - \mu N.
\]
Solving and using  the comparison principle,  for all $t\ge0$, one has
\[
N(t)\le \max\!\left\{N(0),\frac{\Lambda}{\mu}\right\}.
\]
If $N(0)\le\Lambda/\mu$, then $N(t)\le\Lambda/\mu$ for all $t\ge0$,
showing that $\Omega$ is positively invariant. Moreover, for arbitrary
$N(0)\ge0$, $\overline N(t)\to\Lambda/\mu$ as $t\to\infty$, so $N(t)$
is ultimately bounded by $\Lambda/\mu$. This construction is standard
in SIR-type models with vital dynamics
\cite{BrauerCastilloChavez2001,hethcote2000}.
\end{proof}

  \subsection{Constant stress rate}

To obtain explicit analytical results for equilibria and thresholds, we first focus on a baseline case in which water conditions are approximately stationary on the epidemic time scale, so that
\[
\alpha(t) \equiv \alpha \in (0,\alpha_{\max}]
\]
can be treated as a constant.
Under this assumption, \eqref{eq:model-ODE-alpha-t} reduces to
\begin{equation}\label{sys:model-ODE-constant-stress}
\begin{cases}
\displaystyle \frac{dS_N}{dt} = \Lambda - \alpha S_N - \lambda(t) S_N - \mu S_N,\\[6pt]
\displaystyle \frac{dS_S}{dt} = \alpha S_N - \lambda(t) S_S - \mu S_S,\\[6pt]
\displaystyle \frac{dI_N}{dt} = \lambda(t) S_N - (\gamma_N+\mu+d_N) I_N,\\[6pt]
\displaystyle \frac{dI_S}{dt} = \lambda(t) S_S - (\gamma_S+\mu+d_S) I_S,\\[6pt]
\displaystyle \frac{dR}{dt}   = \gamma_N I_N + \gamma_S I_S - \mu R,
\end{cases}
\end{equation}
with $\lambda(t) = \beta_N I_N(t) + \beta_S I_S(t)$.
 Models with multiple susceptibility or risk classes have been widely studied in the epidemiological literature. 
In particular, two-group compartmental models have been used to account for heterogeneity in host behavior or physiology, often leading to nontrivial threshold dynamics \cite{hadeler1997backward,MagalSeydiWebb2016}. 
Related formulations also arise in stage-structured or risk-structured epidemic models, where transitions between host classes play a central role in shaping disease persistence and extinction properties \cite{DjuikemGrognardTouzeau2024}. 
Two-patch or source–sink models with coupling between subpopulations further illustrate how transfers between classes or locations can fundamentally alter equilibrium structure and stability \cite{Li2022MBE}.

Model~\eqref{sys:model-ODE-constant-stress} shares this general two-class structure, but with transition only in the susceptible class and one way, but differs in the interpretation and role of the coupling mechanism.
Here, the transition from the non-stressed to the stressed susceptible class is driven by an environmentally induced stress rate $\alpha$, rather than by spatial migration, behavioural change, or ontogenic progression, especially in the general model~\eqref{eq:model-ODE-alpha-t} with non-constant $\alpha(t)$.

\subsubsection{The basic reproduction number and system behaviour}
\label{sec:constant-stress-R0}
At the disease-free equilibrium (DFE), $I_N = I_S = 0$ and infection cannot invade unless new infections on average exceed recovery and removal processes. 
Then, \eqref{sys:model-ODE-constant-stress} admits a unique DFE given by
\begin{equation}
\label{eq:DFE-constant-stress}
E^0 = \bigg( S_N^0,\, S_S^0,\, 0,\, 0,\, 0 \bigg),
\quad
S_N^0 = \frac{\Lambda}{\alpha + \mu},
\quad
S_S^0 = \frac{\alpha \Lambda}{\mu(\alpha + \mu)}.
\end{equation}

Using the next-generation matrix method \cite{VdDWatmough2002} and denoting $\nu_j=\gamma_j+\mu+d_j$, $j\in\{N,S\}$, the basic reproduction number $\mathcal{R}_0$ is obtained as:
\begin{equation}\label{eq:R0-determ}
\mathcal{R}_0 = 
\underbrace{\frac{\beta_N S_N^0}{\nu_N}}_{\mathcal{R}_{01}}
\;+\;
\underbrace{\frac{\beta_S S_S^0}{\nu_S}}_{\mathcal{R}_{02}}.
\end{equation}

The decomposition in~\eqref{eq:R0-determ} highlights two distinct transmission pathways.
The first term, $\mathcal{R}_{01}$, represents secondary infections generated by the ``normal'' phenotype.
The second term, $\mathcal{R}_{02}$, captures the contribution of the ``stressed'' phenotype.
Since stress is expected to increase shedding ($\beta_S > \beta_N$) and potentially delay recovery ($\nu_S < \nu_N$), a shift in population structure from $S_N$ to $S_S$ (increasing $S_S^0$) effectively shifts the weight of transmission toward the more efficient $\mathcal{R}_{02}$ pathway, thereby increasing the overall epidemic risk.

System~\eqref{sys:model-ODE-constant-stress} also admits a unique endemic equilibrium 
\begin{equation}\label{eq:EEP-constant-stress}
E^* = (S_N^*, S_S^*, I_N^*, I_S^*, R^*),
\end{equation}
where 
\[
\begin{aligned}
S_N^* &= \frac{\Lambda}{\alpha + \mu + \lambda^*}, &
S_S^* &= \frac{\alpha S_N^*}{\mu + \lambda^*}, \\[4pt]
I_N^* &= \frac{\lambda^* S_N^*}{\nu_N}, &
I_S^* &= \frac{\lambda^* S_S^*}{\nu_S}, \; R^*=\frac{\gamma_N I_N^*+\gamma_SI_S^*}{\mu},
\end{aligned}
\]
are expressed in terms of the total force of infection
\[
\lambda^* = \beta_N I_N^* + \beta_S I_S^*.
\]

The behaviour of \eqref{sys:model-ODE-constant-stress} is then governed by the following result, which is proved in several lemmas in Sections~\ref{sec:DFE-constant-stress} and \ref{sec:EEP-constant-stress}.
\begin{proposition}
\label{prop:behaviour-constant-stress}
In the region $\Omega$, when  $\mathcal R_0<1$, \eqref{sys:model-ODE-constant-stress} undergoes a forward transcritical bifurcation at $\mathcal R_0=1$:
\begin{itemize}
    \item when $\mathcal R_0<1$, the disease-free equilibrium $E^0$ is globally asymptotically stable and the endemic equilibrium is not biologically relevant;
    \item for $\mathcal R_0>1$, the disease-free equilibrium $E^0$ is unstable and the unique endemic equilibrium $E^*$ is locally asymptotically stable in $\Omega$.
\end{itemize}
\end{proposition}

\subsubsection{Stability of the disease-free equilibrium}
\label{sec:DFE-constant-stress}
The computation of the basic reproduction number using the next-generation matrix method \cite{VdDWatmough2002} enables us to obtain local asymptotic stability of the disease-free equilibrium when $\mathcal{R}_0<1$ and its instability when $\mathcal{R}_0>1$.

By applying Lyapunov's method for global asymptotic stability, we have the following Lemma.
\begin{lemma}\label{lm:gas-sta}
The DFE $E^0$ of \eqref{sys:model-ODE-constant-stress} is globally asymptotically stable in the biological invariant region $\Omega$ if $\mathcal{R}_0 \le 1$.
\end{lemma}

\begin{proof}
We construct a Lyapunov function $V: \Omega \to \mathbb{R}$. Let us define the constants $c_N = \beta_N/\nu_N$ and $c_S = \beta_S/\nu_S$, where $\nu_j = \gamma_j + \mu + d_j$ for $j=N,S$. Consider the following linear Lyapunov function candidate:
\begin{equation}\label{eq:Lyapunov-func}
V(t) = c_N I_N(t) + c_S I_S(t) = \frac{\beta_N}{\nu_N} I_N + \frac{\beta_S}{\nu_S} I_S.
\end{equation}
Since $I_N, I_S \ge 0$, it follows that $V(t) \ge 0$ for all variables in $\Omega$, and $V(t) = 0$ if and only if $I_N = I_S = 0$.

Differentiating $V$ with respect to time along the trajectories of \eqref{sys:model-ODE-constant-stress} yields
\[
\frac{dV}{dt} = \frac{\beta_N}{\nu_N} \dot{I}_1 + \frac{\beta_S}{\nu_S} \dot{I}_2.
\]
Substituting the expressions for $\dot{I}_1$ and $\dot{I}_2$:
\[
\begin{aligned}
\frac{dV}{dt} &= \frac{\beta_N}{\nu_N} \left[ (\beta_N I_N + \beta_S I_S) S_N - \nu_N I_N \right] + \frac{\beta_S}{\nu_S} \left[ (\beta_N I_N + \beta_S I_S) S_S - \nu_S I_S \right] \\[6pt]
&= \frac{\beta_N}{\nu_N} (\beta_N I_N + \beta_S I_S) S_N - \beta_N I_N + \frac{\beta_S}{\nu_S} (\beta_N I_N + \beta_S I_S) S_S - \beta_S I_S.
\end{aligned}
\]
We factor out the total force of infection $\lambda(t) = \beta_N I_N + \beta_S I_S$:
\[
\begin{aligned}
\frac{dV}{dt} &= (\beta_N I_N + \beta_S I_S) \left( \frac{\beta_N S_N}{\nu_N} \right) - \beta_N I_N + (\beta_N I_N + \beta_S I_S) \left( \frac{\beta_S S_S}{\nu_S} \right) - \beta_S I_S \\[6pt]
&= (\beta_N I_N + \beta_S I_S) \left[ \frac{\beta_N S_N}{\nu_N} + \frac{\beta_S S_S}{\nu_S} \right] - (\beta_N I_N + \beta_S I_S) \\[6pt]
&= \lambda(t) \left[ \left( \frac{\beta_N S_N}{\nu_N} + \frac{\beta_S S_S}{\nu_S} \right) - 1 \right].
\end{aligned}
\]
Recall that in the invariant region $\Omega$, the susceptible populations are bounded asymptotically by their disease-free values, i.e., $S_N(t) \le S_N^0$ and $S_S(t) \le S_S^0$. Therefore, we have:
\[
\frac{\beta_N S_N}{\nu_N} + \frac{\beta_S S_S}{\nu_S} \;\le\; \frac{\beta_N S_N^0}{\nu_N} + \frac{\beta_S S_S^0}{\nu_S} = \mathcal{R}_0.
\]
Using this inequality, the derivative of the Lyapunov function satisfies:
\begin{equation}
\frac{dV}{dt} \le \lambda(t) (\mathcal{R}_0 - 1).
\end{equation}
If $\mathcal{R}_0 \le 1$ and considering $\lambda(t) \ge 0$, it follows that $\frac{dV}{dt} \le 0$ for all $t$.

Furthermore, $\frac{dV}{dt} = 0$ implies that either $\lambda(t) = 0$ (which implies $I_N=I_S=0$) or $\mathcal{R}_0 = 1$ and $S_N=S_N^0, S_S=S_S^0$. In either case, the largest invariant set contained in the set where $\dot{V}=0$ is the singleton $\{E^0\}$. By the LaSalle Invariance Principle \cite{lasalle1976stability}, every solution starting in $\Omega$ approaches $E^0$ as $t \to \infty$. Thus, $E^0$ is globally asymptotically stable.
\end{proof}

\subsubsection{Endemic equilibrium and its stability}
\label{sec:EEP-constant-stress}
Let $E^* = (S_N^*, S_S^*, I_N^*, I_S^*, R^*)$ be an  endemic equilibrium of \eqref{sys:model-ODE-constant-stress}. 
Defining the total force of infection as
\[
\lambda^* = \beta_N I_N^* + \beta_S I_S^*,
\]
we express all steady-state values in terms of $\lambda^*$:
\[
\begin{aligned}
S_N^* &= \frac{\Lambda}{\alpha + \mu + \lambda^*}, &
S_S^* &= \frac{\alpha S_N^*}{\mu + \lambda^*}, \\[4pt]
I_N^* &= \frac{\lambda^* S_N^*}{\nu_N}, &
I_S^* &= \frac{\lambda^* S_S^*}{\nu_S}, \; R^*=\frac{\gamma_N I_N^*+\gamma_SI_S^*}{\mu}.
\end{aligned}
\]
Substituting these relations into the definition of $\lambda^*$ yields a polynomial
\begin{equation}\label{eq:lambda-poly}
P(\lambda^*)=A_2 \lambda^{*2} + A_1 \lambda^* + A_0 = 0,
\end{equation}
where the coefficients are
\[
\begin{aligned}
A_2 &= \nu_N\nu_S, \\[4pt]
A_1 &= \nu_N\nu_S\mu+\nu_N\nu_S(\alpha + \mu)(1 - \mathcal{R}_{01}), \\[4pt]
A_0 &= \nu_N\nu_S\mu(\alpha + \mu)(1 - \mathcal{R}_0).
\end{aligned}
\]

The sign structure of $(A_2, A_1, A_0)$ determines whether a positive root $\lambda^*>0$ exists. 
Applying Descartes sign rules for polynomial $P(\lambda^*)$ , the number of positive
solutions is given in Table~\ref{tab:descartes}.

\begin{table}[h!]
\centering
\caption{Descartes rule for $\lambda^*$-polynomial~\eqref{eq:lambda-poly}.}
\label{tab:descartes}
\begin{tabular}{ccccccl}
\hline
\textbf{Case} & $\mathcal{R}_{01}$ & $\mathcal{R}_0$ & $\operatorname{sign}(A_2)$ & $\operatorname{sign}(A_1)$ & $\operatorname{sign}(A_0)$ & Number of positive solutions \\
\hline
1 & $<1$ & $<1$ & $+$ & $+$ & $+$ & No solution\\
2 & $<1$ & $>1$ & $+$ & $+$ & $-$ & One solution\\
3 & $>1$ & $>1$ & $+$ & $-$ & $-$ & One solution\\
\hline
\end{tabular}
\end{table}
The polynomial equation~\eqref{eq:lambda-poly} admits a unique positive root $\lambda^*>0$ if and only if $\mathcal{R}_0>1$. With the positive value of $\lambda^*$, $E^*$ is positive. 
Then, we obtain the following Lemma about the existence of the endemic equilibrium
\begin{lemma}
If $\mathcal{R}_0>1$, \eqref{sys:model-ODE-constant-stress} admits a unique endemic equilibrium $E^*$.
\end{lemma}

To finish the proof of Proposition~\ref{prop:behaviour-constant-stress}, we investigate the local stability of the endemic equilibrium of \eqref{sys:model-ODE-constant-stress}.
Let us consider that $\beta_N=\beta$ and $\beta_S=\sigma\beta$, then the basic reproduction number can be written as
\[
\mathcal R_0(\beta) = \beta\,\kappa, \qquad
\kappa = \frac{S_N^0}{\nu_N}
         + \frac{\sigma S_S^0}{\nu_S}.
\]
Hence the threshold condition $\mathcal R_0=1$ corresponds to
\begin{equation}\label{eq:beta_star}
    \beta^* = \frac{1}{\kappa} =\frac{\nu_N\nu_S}{\nu_NS_N^0+\nu_N\sigma S_S^0}.
\end{equation}

We now apply \cite[Theorem~4.1]{castillo2004dynamical} with the transmission
coefficient $\beta$ as bifurcation parameter.  The proof is given in Appendix~\ref{app:proof-LAS-EE}.

\subsection{Time-dependent stress}

\subsubsection{Disease-free trajectory and reproduction number}

When infection is absent ($I_N = I_S = 0$), system~\eqref{eq:model-ODE-alpha-t} reduces to a non-autonomous subsystem for the susceptible classes,
\begin{equation}\label{eq:DFE-trajectory}
\begin{cases}
\displaystyle \frac{dS_N}{dt} = \Lambda - \big(\alpha(t) + \mu\big) S_N,\\[6pt]
\displaystyle \frac{dS_S}{dt} = \alpha(t) S_N - \mu S_S,
\end{cases}
\end{equation}
with $I_N^0(t) = I_S^0(t) = R^0(t) = 0$.
Thus, in contrast to the constant-stress case, the disease-free state is no longer a fixed point, but a time-dependent trajectory
\[
E^0(t)=  \big(S_N^0(t), S_S^0(t), 0, 0, 0\big).
\]
where $(S_N^0(t), S_S^0(t))$ is the solution of the solution of~\eqref{eq:DFE-trajectory}.

To formally analyze the stability, we write the dynamics of the infected compartments $I(t) = (I_N(t), I_S(t))^T$ near the disease-free trajectory of \eqref{eq:model-ODE-alpha-t}. 
The linearized system takes the form of a non-autonomous linear differential equation:
\begin{equation}\label{eq:matrix-system}
\frac{dI}{dt} = \big( F(t) - V \big) I(t) := J(t) I(t),
\end{equation}
where
\[
V = \begin{pmatrix} \nu_N & 0 \\ 0 & \nu_S \end{pmatrix}, \quad
F(t) = \begin{pmatrix} \beta_N S_N^0(t) & \beta_S S_N^0(t) \\ \beta_N S_S^0(t) & \beta_S S_S^0(t) \end{pmatrix}.
\]
The matrix $J(t) = F(t) - V$ is time-dependent and its off-diagonal entries are non-negative (since $\beta_N, \beta_S, S_{N}^0(t),S_{S}^0(t \ge 0$). Such a matrix is called a \emph{Metzler matrix}, which implies that the system is \emph{cooperative}. This property allows us to use the standard Comparison Theorem for differential equations~\cite{SmithMonotoneSystems}.

We define the spectral radius of a constant matrix $M$ as $s(M) := \max \{ \text{Re}(\lambda) : \lambda \in \sigma(M) \}$.
By analogy with the autonomous case, this suggests the definition of a \emph{time-dependent  reproduction number}
\begin{equation}\label{eq:Rt-def}
\mathcal{R}(t):=\mathcal R(\alpha(t))
= \frac{\beta_N S_N^0(t)}{\nu_N}
+ \frac{\beta_S S_S^0(t)}{\nu_S}.
\end{equation}

We now assume that the time-varying stress rate is uniformly bounded,
\begin{equation}\label{eq:alpha-bounds}
0\leq \alpha_{\min} \le \alpha(t) \le \alpha_{\max}
\qquad \text{for all } t \ge 0,
\end{equation}
for some fixed constants $\alpha_{\min}$ and $\alpha_{\max}$.
For each constant value $\alpha$, the autonomous system~\eqref{sys:model-ODE-constant-stress} admits a basic reproduction number
\begin{equation}\label{eq:R0-of-alpha}
\mathcal{R}_0(\alpha)
= \frac{\Lambda}{\alpha + \mu}
\left(
\frac{\beta_N}{\nu_N}
+ \frac{\alpha\,\beta_S}{\mu\,\nu_S}
\right),
\end{equation}
Biologically, stressed fish are assumed to be at least as infectious or to remain infectious for longer than non-stressed fish, which we express as
\begin{equation}\label{eq:stress-bio-assumption}
\frac{\beta_S}{\nu_S} \;\ge\; \frac{\beta_N}{\nu_N}.
\end{equation}
Under this assumption, one has
\[
\mathcal{R}_0(\alpha_{\min})
\;\le\;
\mathcal{R}_0(\alpha(t))
\;\le\;
\mathcal{R}_0(\alpha_{\max})
\qquad \text{for all } t,
\]

Indeed, for each constant value $\alpha$, the autonomous system \eqref{sys:model-ODE-constant-stress} has basic reproduction number $\mathcal R_0(\alpha)$ given by \eqref{eq:R0-of-alpha}, which can be written
\[
\mathcal{R}_0(\alpha)
= \Lambda\,\frac{A + \alpha B}{\alpha + \mu},
\qquad
A = \frac{\beta_N}{\nu_N},\quad
B = \frac{\beta_S}{\mu\,\nu_S},
\]
and differentiating with respect to $\alpha$ yields
\[ \frac{\partial  \mathcal R_0 (\alpha) }{\partial \alpha}
= \Lambda\,\frac{B(\alpha+\mu) - (A+\alpha B)}{(\alpha+\mu)^2}
= \Lambda\,\frac{B\mu - A}{(\alpha+\mu)^2}.
\]
Since $B\mu = \beta_S/\nu_S$, the biological condition
\eqref{eq:stress-bio-assumption} implies $B\mu \ge A$, hence
$ \frac{\partial  \mathcal R_0 (\alpha) }{\partial \alpha}\ge 0$ which implies that $\mathcal R_0(\alpha)$ increases as a function of $\alpha$.
In particular, if $0\le \alpha_{\min} \le \alpha(t) \le \alpha_{\max}$
for all $t\ge 0$, then
\[ \mathcal R_{01}=\mathcal R_0(0)\le
\mathcal R_0(\alpha_{\min})
\;\le\;
\mathcal R_0(\alpha(t))
\;\le\;
\mathcal R_0(\alpha_{\max})
\qquad \text{for all } t\ge 0.
\]

Under these assumptions, the time-dependent reproduction number $\mathcal R(t)$ remains bounded between the two autonomous thresholds associated with the extremal stress rates.
To illustrate this, we can compare several stress profiles $\alpha(t)$ with the corresponding trajectories of $\mathcal R(t)$ along the disease-free solution $E^0(t)$: for each scenario, $\mathcal R(t)$ fluctuates in time but never crosses the lower bound $\mathcal R_0(\alpha_{\min})$ nor the upper bound $\mathcal R_0(\alpha_{\max})$.

Figure~\ref{fig:Rt-bounds} shows the dynamics of $\mathcal R(t)$ using the ODE~\eqref{eq:DFE-trajectory} at the DFE. When the ODE converges, the figure shows the boundedness of this time-dependent basic reduction number, as shown in equation~\eqref{eq:Rt-def}. The parameter values used for the rest of the simulations are \(\mu=0.002\), \(d_N=0.01\), and \(d_S=0.01\) (see Table~\ref{tab:parameters}); the remaining parameters are as in Table~\ref{tab:parameters}.

\begin{figure}[H]
    \centering
    \includegraphics[width=0.8\linewidth]{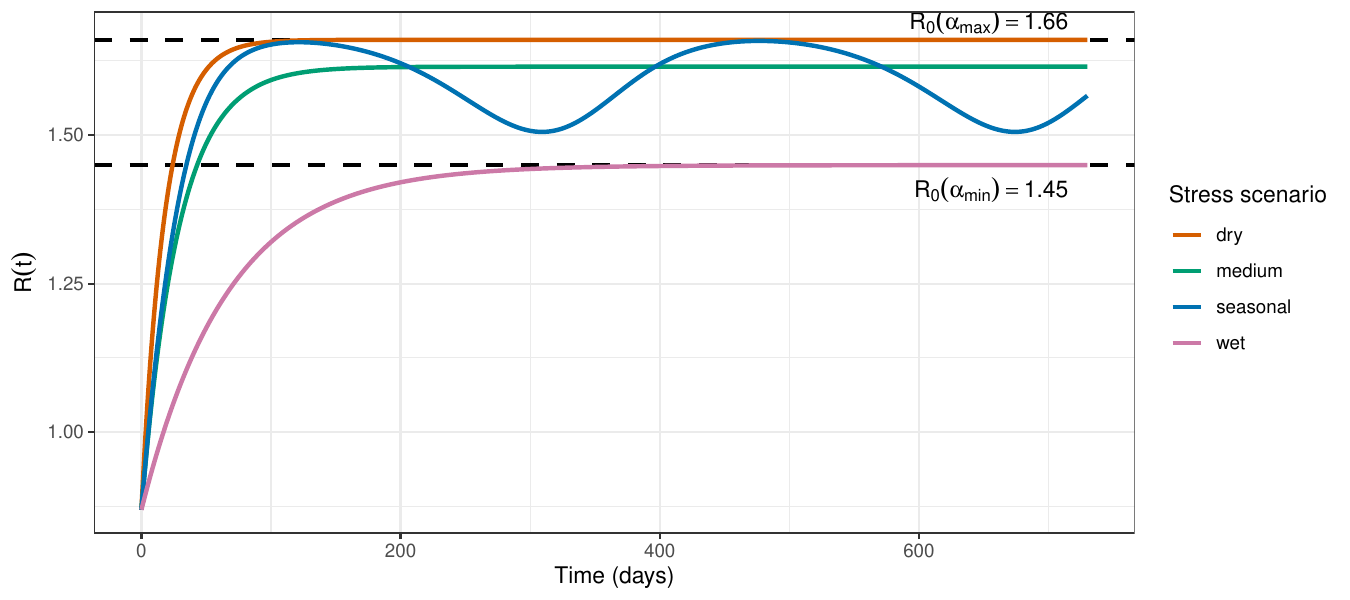}
    \caption{Time-dependent reproduction number $\mathcal R(t)$ under bounded stress.
    Each coloured curve represents $\mathcal R(t)$ for a different stress profile $\alpha(t)$ satisfying $\alpha_{\min} \le \alpha(t) \le \alpha_{\max}$.
    The two horizontal lines correspond to the autonomous thresholds $\mathcal R_0(\alpha_{\min})$ and $\mathcal R_0(\alpha_{\max})$.}
    \label{fig:Rt-bounds}
\end{figure}

\subsubsection{Endemic equilibrium and its stability}

A fundamental result in mathematical epidemiology links the basic reproduction number to this spectral abscissa \cite[Proof~Theorem 2]{VdDWatmough2002}:
\begin{equation}\label{eq:spectral-relation}
\mathcal{R}_0(\alpha) < 1 \iff s\big(F(\alpha) - V\big) < 0.
\end{equation}
We have the following result about the existence and the persistence of the disease.
\begin{lemma}
\label{lem:stress-bounds}
Assume that $\alpha(t)$ satisfies~\eqref{eq:alpha-bounds} and that biological condition~\eqref{eq:stress-bio-assumption} holds.
\begin{itemize}
    \item[\emph{(i)}] If $\mathcal{R}_0(\alpha_{\max}) < 1$, the disease-free trajectory is globally asymptotically stable. In particular, the disease-free trajectory of~\eqref{eq:model-ODE-alpha-t} is uniformly asymptotically stable.
    \item[\emph{(ii)}] If $\mathcal{R}_0(\alpha_{\min}) > 1$, the disease-free trajectory is unstable and the infection can invade even under the low-stress regime.
\end{itemize}
\end{lemma}

\begin{proof}
\noindent\emph{(i)}
Let $\bar{\alpha} = \alpha_{\max}$. Since the reproduction number $\mathcal{R}_0(\alpha)$ is monotonically increasing with respect to stress, the condition $\mathcal{R}_0(\alpha_{\max}) < 1$ implies that the worst-case autonomous matrix $\bar{J} = F(\bar{\alpha}) - V$ satisfies the stability condition:
\[
s(\bar{J}) < 0.
\]
Under the assumption that stress increases overall transmission potential, the time-dependent matrix $J(t)$ is bounded by the worst-case constant matrix $\bar{J}$:
\[
J(t) \le \bar{J} \quad \text{for all } t \ge 0.
\]
By the Comparison Theorem for cooperative systems \cite{SmithMonotoneSystems}, the solution $I(t)$ of system~\eqref{eq:matrix-system} is bounded by the solution $Y(t)$ of the constant coefficient system $\frac{dY}{dt} = \bar{J} Y$ with $Y(0) = I(0)$:
\[
0 \le I(t) \le Y(t) = e^{\bar{J}t} I(0).
\]
Since $s(\bar{J}) < 0$, the matrix exponential $e^{\bar{J}t} \to 0$ as $t \to \infty$. Since $(I_N(t),I_S(t)) \to (0,0)$, by replacing into~\eqref{eq:model-ODE-alpha-t} we obtain~\eqref{eq:DFE-trajectory}, which implies that $(S_N(t),S_S(t)) \to (S_N^0(t),S_S^0(t))$. Thus, $E^0(t)$ is global asymptotic stability.

\medskip
\noindent\emph{(ii)}
Conversely, let $\underline{\alpha} = \alpha_{\min}$. If $\mathcal{R}_0(\alpha_{\min}) > 1$, then by monotonicity $\mathcal{R}_0(\alpha(t)) > 1$ for all $t$.
This implies that the spectral bound of the instantaneous matrix is strictly positive:
\[
s(J(t)) \ge s(F(\underline{\alpha}) - V) > 0.
\]
Since the system~\eqref{eq:matrix-system} is linear and cooperative, and the spectral radius is uniformly positive, the origin is unstable. Any small perturbation $I(0) > 0$ will grow initially, implying persistence of the infection. Then, there exists at least one admissible stress
trajectory for which persistence occurs
\end{proof}

Figure~\ref{fig:ODE-I1-I2-alpha-t} illustrates how the four water–stress
scenarios shape the epidemic profiles in the deterministic model.
In all cases, the stressed infected class $I_S$ reaches
higher peaks than the non-stressed class $I_N$, reflecting the higher
transmission rates in stressed fish.
Dry and medium conditions, which correspond to persistently elevated
stress, generate the largest and earliest epidemic peaks, whereas the
wet scenario produces much smaller and later outbreaks in both $I_N$ and $I_S$,
consistent with a reproduction number close to the invasion threshold.
Under seasonal stress, infections occur in a sequence of waves whose
timing follows the oscillations in the stress rate $\alpha(t)$.
These simulations highlight that environmental stress does not merely
change the magnitude of infection, but also reshapes its timing and
distribution between non-stressed and stressed hosts.

\begin{figure}[H]
    \centering
    \includegraphics[width=0.8\textwidth]{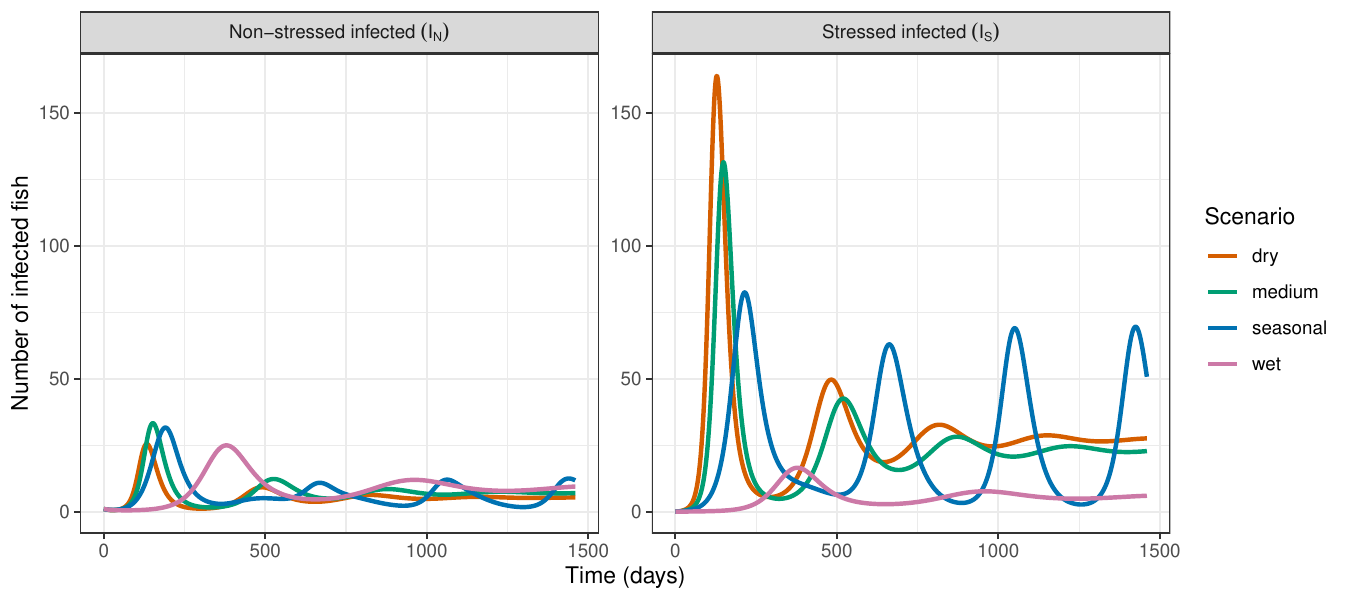}
    \caption{Deterministic trajectories for the stress–infection model with time-varying stress rate $\alpha(t)$ with different basic reproduction numbers in Figure~\ref{fig:Rt-bounds}.
    The left panel shows the number of non-stressed infected fish $I_N(t)$, the right panel shows the number of stressed infected fish $I_S(t)$, under four water-stress scenarios (dry, medium, seasonal, and wet).}
    \label{fig:ODE-I1-I2-alpha-t}
\end{figure}

\section{Stochastic model and analysis} \label{sec:stochastic}

\subsection{Continuous-time Markov chain }

In this subsection, we use a Continuous-time Markov chain (CTMC) to determine the first time introduction of the disease in different groups of fish.

Ordinary differential equation models are the limit of a continuous-time Markov chain (CTMC) when the population and the number of realisations become large \cite{Kurtz1970}.
To capture stochastic fluctuations inherent in finite populations, we therefore formulate the CTMC associated with the deterministic system \eqref{sys:model-ODE-constant-stress}.
Let
\begin{equation}\label{eq:CTMC}
\mathbf{X}(t) = \big(S_N(t), S_S(t), I_N(t), I_S(t), R(t)\big),
\quad t \ge 0,
\end{equation}
denote the vector of population counts in each compartment at time $t$.
Each transition in the system corresponds to a discrete jump in $\mathbf{X}_t$ governed by an infinitesimal rate, or \emph{propensity}, which specifies the probability per unit time that a given event occurs.

The corresponding infinitesimal transition probabilities are defined as
\begin{equation}\label{eq:CTMC-transition}
\mathbb{P}\!\left(\mathbf{X}(t+\Delta t)=\mathbf{x}' \,\middle|\, \mathbf{X}(t)=\mathbf{x}\right)
= q(\mathbf{x},\mathbf{x}') \Delta t + o(\Delta t),
\end{equation}
where $q(\mathbf{x},\mathbf{x}')$ is the transition rate from state $\mathbf{x}$ to state $\mathbf{x}'$  given in the Table~\ref{tab:CTMC-stress}.

\begin{table}[h!]
\centering
\caption{Transitions and corresponding rates for the CTMC representation of the stress–infection model.}
\label{tab:CTMC-stress}
\begin{tabular}{lll}
\hline
\textbf{Transition} & \textbf{State change} & \textbf{Rate } \\
\hline
Birth & $S_N \to S_N + 1$ & $\Lambda$ \\[2pt]
Natural death of $x$ & $x\to x - 1$ & $\mu x$ \\
Stress transition & $S_N \to S_N - 1,\, S_S \to S_S + 1$ & $\alpha(t) S_N$ \\[2pt]
Infection of $S_N$ & $S_N \to S_N - 1,\, I_N \to I_N + 1$ & $(\beta_N I_N + \beta_S I_S) S_N$ \\[2pt]
Infection of $S_S$ & $S_S \to S_S - 1,\, I_S \to I_S + 1$ & $(\beta_N I_N + \beta_S I_S) S_S$ \\[2pt]
Recovery from $I_N$ & $I_N \to I_N - 1,\, R \to R + 1$ & $\gamma_N I_N$ \\[2pt]
Recovery from $I_S$ & $I_S \to I_S - 1,\, R \to R + 1$ & $\gamma_S I_S$ \\[2pt]
Disease-induced death of $I_N$ & $I_N \to I_N - 1$ & $d_N I_N$ \\[2pt]
Disease-induced death of $I_S$ & $I_S \to I_S - 1$ & $d_S I_S$ \\
\hline
\end{tabular}
\end{table}

\subsection{Branching process approximation}

During the early phase of an epidemic, when the number of infections is small, the CTMC can be approximated by a multitype branching process approximation (MBPA). 
In this phase, the susceptible population remains close to its disease-free equilibrium value $(S_N^0,S_S^0)$ and infections occur approximately independently.
Let
\[
\mathbf{I}(t) = (I_N(t), I_S(t)), \quad t \ge 0,
\]
denote the MBPA corresponding to infected classes.
Each infected individual of type $j$ produces new infections and terminates by recovery or death according to exponential rates derived from Table~\ref{tab:CTMC-stress}.
The reproduction kernel  of the branching process is given by
\[
\mathbf{M} =
\begin{pmatrix}
m_{11} & m_{12} \\
m_{21} & m_{22}
\end{pmatrix}
=
\begin{pmatrix}
\displaystyle \frac{\beta_N S_N^0}{\nu_N} & \displaystyle \frac{\beta_N S_S^0}{\nu_N} \\[8pt]
\displaystyle \frac{\beta_S S_N^0}{\nu_S} & \displaystyle \frac{\beta_S S_S^0}{\nu_S}
\end{pmatrix}.
\]
Entry $m_{jk}$ represents the expected number of offspring of type $k$ produced by one individual of type $j$ during its entire infectious period.
The spectral radius $\rho(\mathbf{M})$ is exactly the basic reproduction number $\mathcal{R}_0$.

The offspring probability generating functions (p.g.f.) corresponding to the offspring distribution per individual lifetime are
\begin{equation}\label{eq:pgf-stress} 
\left\{
\begin{aligned}
f_N(u_N, u_S) &= \frac{ \beta_N S_N^0 u_N^2 + \beta_N S_S^0 u_N u_S+\nu_N}{\nu_N + \beta_N(S_N^0+S_S^0)}, \\
f_S(u_N, u_S) &= \frac{\beta_S S_S^0 u_S^2 + \beta_S S_N^0 u_N u_S + \nu_S }{\nu_S + \beta_S(S_N^0+S_S^0)}.
\end{aligned}
\right.
\end{equation}

Let $\mathbf{u}^*=(u_N^*,u_S^*)$ be a fixed point of the map $\mathbf{f}(\mathbf{u})=(f_N(\mathbf{u}),f_S(\mathbf{u}))$ satisfying $\mathbf{f}(\mathbf{u}^*)=\mathbf{u}^*$. We have the following theorem.

\begin{theorem}[Probability of extinction]
\label{thm:extinction-stress}
The multitype branching process associated with the CTMC~\eqref{eq:CTMC-transition} is positive, regular and nonsingular. 
Let the initial number of infected individuals be $(I_N(0),I_S(0)) = (i_{N0},i_{S0})$.
Then, the probability of extinction of the infection is
\begin{equation}
\mathbb{P}_{\text{ext}} = (u_N^*)^{i_{N0}} (u_S^*)^{i_{S0}},
\end{equation}
where $\mathbf{u}^*$ is the smallest fixed point of $\mathbf{f}(\mathbf{u})$ in $[0,1]^2$.
Furthermore,
\begin{itemize}
\item if $\mathcal{R}_0 \le 1$, then $\mathbf{u}^* = (1,1)$ and $\mathbb{P}_{\text{ext}} = 1$;
\item if $\mathcal{R}_0 > 1$, then there exists a unique $\mathbf{u}^* \in (0,1)^2$ such that $\mathbf{f}(\mathbf{u}^*)=\mathbf{u}^*$, and hence $\mathbb{P}_{\text{ext}} < 1$.
\end{itemize}
\end{theorem}

\begin{proof}
We prove that the associated multitype branching process is positive, regular and nonsingular
and that its extinction properties are determined by the spectral radius of the first-moment matrix.

\begin{itemize}
    \item[(i)] \textbf{Nonnegativity and monotonicity.}
    For all $j,k \in  \{N,S\}$, compute the partial derivative of~\eqref{eq:pgf-stress};
    for all $\mathbf{u}\in [0,1]^2$,
    \[
    \frac{\partial f_j}{\partial u_j}(\mathbf{u})
    = \frac{\beta_j (2S_j^0u_j+S_k^0u_k)}{\beta_j (S_N^0 + S_S^0)+ \nu_j}
    \text{ and }   \frac{\partial f_j}{\partial u_k}(\mathbf{u})
    = \frac{\beta_j S_k^0 u_j}{\beta_j (S_N^0 + S_S^0)+ \nu_i}.
    \]
    Hence the Jacobian matrix $Df(\mathbf{u})$ is entrywise nonnegative and the map 
    $\mathbf{f}$ is monotone on $[0,1]^2$.
    By \cite[Theorem~2.3, p.~113]{berman1979nonnegative}, the corresponding multitype branching process is \emph{nonsingular}.

    \item[(ii)] \textbf{First-moment matrix and primitivity.}
    The mean (first-moment) matrix is given by
    \[
    \mathbf{M} = Df(\mathbf{1}), \quad \mathbf{1} = (1,1)^{\top}.
    \]
    Since the derivatives of $f_j$ are constant in $\mathbf{u}$,
    we obtain
    \[M_{jj}
    = \frac{\beta_j(2 S_j^0+S_k^0)}{\beta_j (S_N^0 + S_S^0)+ \nu_j} \text{ and }
    M_{jk}
    = \frac{\beta_j S_k^0}{\beta_j (S_N^0 + S_S^0)+ \nu_j},
    \qquad j,k \in  \{N,S\}.
    \]
    All entries of $\mathbf{M}$ are strictly positive if $\beta_k S_k^0 > 0$,
    hence $\mathbf{M}$ is a positive matrix and thus \emph{primitive} matrix 
    
\end{itemize}

From (i) and (ii), the associated multitype branching process is 
\emph{positive, regular, and nonsingular}.
Let $\rho(\mathbf{M})$ denote the spectral radius of $\mathbf{M}$.
By standard branching-process results
\cite{ALLEN201399,Harris1963}, the following holds:
\begin{itemize}
    \item if $\rho(\mathbf{M}) = \mathcal{R}_0 \le 1$, 
    then the only fixed point of $\mathbf{f}$ in $[0,1]^2$ is $\mathbf{1}$,
    implying certain extinction ($\mathbb{P}_{\mathrm{ext}}=1$);
    \item if $\rho(\mathbf{M}) = \mathcal{R}_0 > 1$, 
    then there exists a unique nontrivial fixed point 
    $\mathbf{u}^* \in (0,1)^2$ satisfying $\mathbf{f}(\mathbf{u}^*) = \mathbf{u}^*$,
    leading to a probability of extinction:
    \[
    \mathbb{P}_{\mathrm{ext}} 
    = (u_N^*)^{i_{N0}} (u_S^*)^{i_{S0}}.
    \]
\end{itemize}
This completes the proof.
\end{proof}
This stochastic framework quantifies extinction probabilities and the likelihood of invasion in finite populations.
It complements the deterministic analysis by providing a probabilistic interpretation of the threshold $\mathcal{R}_0=1$.

To illustrate how stress and the seeding class affect the probability of
extinction, we compute
$\mathbb{P}_{\text{ext}}$
for a range of initial conditions and for the four water–stress
scenarios (wet, medium, seasonal, dry), corresponding to increasing
values of $\alpha$.  
In Figure~\ref{fig:Pext-violin-alpha-seed}, the violins show that
extinction is very likely in the wet, low–stress regime
($\alpha = 0.001$), with $\mathbb{P}_{\text{ext}}$ concentrated near~1,
and becomes progressively less likely as $\alpha$ increases:
for the dry scenario ($\alpha = 0.05$), most mass is near~0, indicating
a high chance of a major outbreak.  
For a given $\alpha$, extinction is slightly less likely when the
epidemic starts in the stressed class $I_S$ than when it starts in the
non–stressed class $I_N$, reflecting the higher transmission potential
of stressed fish.

\begin{figure}[H]
    \centering
    \includegraphics[width=\textwidth]{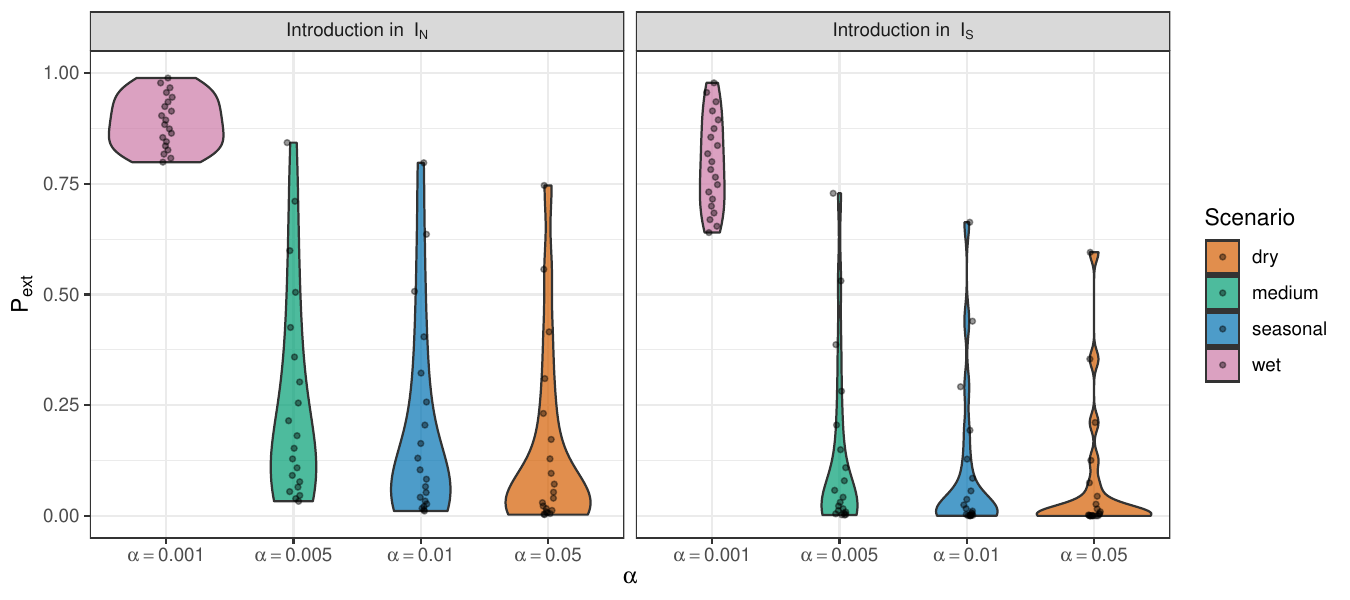}
    \caption{
    Distribution of extinction probability $\mathbb{P}_{\text{ext}}$
    over initial conditions, for four stress levels $\alpha$
    (wet, medium, seasonal, dry) and two seeding classes.
    Left: the epidemic is seeded in $I_N$; right: it is seeded in $I_S$.
    }
    \label{fig:Pext-violin-alpha-seed}
\end{figure}

\subsection{First time introduction of the disease}

In order to quantify how stress modifies the timing of infection, we use the CTMC formulation to compute the \emph{first infection times} in each infected class.
For a given parameter set, we simulate a large number of independent sample paths of the CTMC starting from a disease-free population at its deterministic DFE $(S_N^0,S_S^0,0,0,0)$, and introduce a single infected individual in one of the susceptible classes.
For each trajectory, we record the random times
\begin{equation}\label{eq:tau}
\tau_N = \inf\{t>0 : I_N(t)>0\},
\qquad
\tau_S = \inf\{t>0 : I_S(t)>0\},
\end{equation}
which correspond, respectively, to the first appearance of infection among non-stressed fish ($I_N$) and stressed fish ($I_S$).

Figure~\ref{fig:first-intro-ctmc} shows the distribution of first cross-introduction times obtained from $50,000$ CTMC simulations started from $S_N(0)=10,000$, $S_S(0)=0$ and a single infected fish.
When the epidemic is seeded in the non-stressed class $I_N$ (top panel), the first appearance in $I_S$ is typically delayed by several days: introductions occur earlier under dry (high-stress) conditions and later under wet (low-stress) conditions, with the medium and seasonal scenarios lying in between and exhibiting a longer tail of late introductions.
In contrast, when the epidemic is seeded in the stressed class $I_S$ (bottom panel), the first appearance in $I_N$ occurs almost immediately in all scenarios, with densities sharply concentrated near $0$.

This asymmetry reveals a critical vulnerability.
When infection enters via a stressed individual ($I_S$), the high transmission rate $\beta_S$ allows for an immediate explosive invasion into the abundant $S_N$ class (Figure~\ref{fig:first-intro-ctmc}, bottom).
Conversely, an introduction via $I_N$ faces a stochastic barrier: the infection must survive long enough in the lower-transmission $I_N$ class to either amplify or transmit to a stressed individual.
This delay in the top panel represents a transient window of opportunity for control that is lost if the index case is a stressed fish.

\begin{figure}[H]
  \centering
  \includegraphics[width=\linewidth]{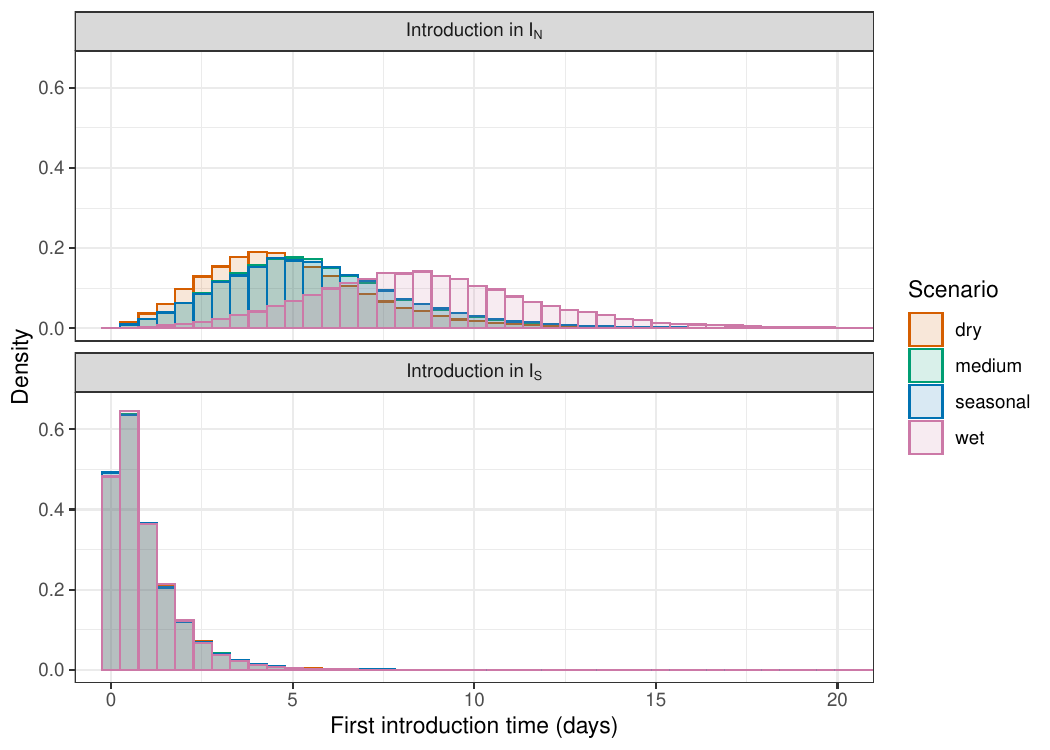}
  \caption{\textbf{First cross-introduction times by seeding class and stress scenario.}
 Top: time to the first appearance ($\tau_S$) in $I_S$ when seeding one infected in $I_N$ at $t=0$.
  Bottom: time to the first appearance ($\tau_N$) in $I_N$ when seeding one infected in $I_S$.
  Colours indicate the four different stress scenarios.}
  \label{fig:first-intro-ctmc}
\end{figure}

\section{Discussion}\label{sec:discussion}

The interactions between environmental stressors and pathogen transmission are complex and often non-linear. While empirical studies have long established that stress compromises fish immunity \citep{snieszko1974stress,tort2011stress}, quantifying how this individual-level physiology scales up to population-level epidemics has remained a challenge. In this study, we developed a model to bridge this gap, treating susceptibility not as a fixed trait but as a dynamic state driven by water quality.

For the deterministic model, we demonstrate that the basic reproduction number $\mathcal{R}_0$ remains a valid threshold for long-term eradication and we quantify how elevated stress levels amplify the outbreak's peak. 
In contrast, stochastic simulations reveal a critical vulnerability: the probability of an outbreak depends strongly on \emph{which} physiological class first introduces the pathogen. 
This finding highlights that in high-stress environments, the effective window for intervention may close far more rapidly than a deterministic model.

Our deterministic analysis reveals that despite the added complexity of host heterogeneity, the system retains a classical forward bifurcation at $\mathcal{R}_0=1$. This is a significant structural property. In many heterogeneous models particularly those with assortative mixing or multi-group structures, backward bifurcations often arise, creating bistable regimes where disease control becomes  difficult \citep{gumel2012causes,hadeler1997backward}. The absence of such bistability in our model suggests that standard threshold-based control strategies remain valid: if the basic reproduction number is less than one, theoretical eradication is possible.

However, our stochastic analysis highlights a critical limitation of relying solely on deterministic thresholds. We identified a ``stochastic barrier'' effect that depends heavily on the physiological state of the index case. When infection is introduced via a normal host ($I_N$), the lower transmission potential ($\beta_N$) and the possibility of recovery before stress induction create a high probability of stochastic extinction. This aligns with the concept of ``evolutionary suicide'' or ``stochastic fade-out'' observed in low-density populations \citep{lloyd2005superspreading}.
In contrast, introduction via a stressed host ($I_S$) effectively bypasses this barrier. Because stressed fish act as ``super-shedders'' with higher transmission rates ($\beta_S$), the infection chain creates a rapid creation of secondary cases, pushing the system quickly out of the stochastic drift phase and into exponential growth. This temporal asymmetry implies that in aquaculture settings, the \emph{timing} of biosecurity breaches relative to water quality events is paramount. A pathogen introduction during a hypoxic event is qualitatively different from one during normoxia, even if the long-term average $\mathcal{R}_0$ is identical.

From a management perspective, our results suggest that monitoring dissolved oxygen (DO) is as critical as monitoring the pathogen itself. Current protocols often focus on quarantining symptomatic fish \citep{murray2002precautionary}. However, our model indicates that the stressed susceptible class ($S_S$) is a silent reservoir of risk.
Interventions that improve water quality do more than just improve fish welfare; they actively restore the stochastic barrier against invasion.
Furthermore, the delay mechanism observed in the introduction in $I_N$ scenarios provides a theoretical window of opportunity. If managers can detect the pathogen while it is still confined to the healthy sub-population, aggressive culling or treatment may succeed. Once the infection bleeds into the stressed class, the window closes.

We assumed a homogeneous mixing assumption within the population, which is standard for compartmental models. However, spatial gradients of oxygen often exist in large water \citep{solstorm2016dissolved}, potentially creating localized hotspots of stressed fish where infection could foster. Future work could extend this framework to a Partial Differential Equation (PDE) model to account for spatial oxygen diffusion. Additionally, while we modeled stress as a one-way transition driven by water quality, fish can recover from stress if conditions improve using control methods. Incorporating a reversible recovery rate $S_S \to S_N$ would allow for the study of resilience strategies in recirculating aquaculture systems (RAS).

\section*{Acknowledgements}
JA acknowledges partial support from NSERC. CD thanks Stéphanie Portet for the supportive environment at the University of Manitoba and acknowledges partial support from a 2024/2025 Maud Menten Institute Research Accelerator Award.

\appendix

\section{End of proof of Proposition~\ref{prop:behaviour-constant-stress}}\label{app:proof-LAS-EE}

\begin{proof}[End of proof of Proposition~\ref{prop:behaviour-constant-stress}]
To prove that the endemic equilibrium is locally asymptotically stable, we apply \cite[Theorem 4.1]{castillo2004dynamical}.
For this, we rewrite \eqref{sys:model-ODE-constant-stress} in the form
\begin{equation}\label{Castillo-flower}
  \frac{dx}{dt} = f(x,\psi),
\end{equation}
with $x\in\mathbb R^5$ and a scalar bifurcation parameter
$\psi\in\mathbb R$, such that $x=0$ is an equilibrium of
\eqref{Castillo-flower} for all $\psi$.

\medskip\noindent
\textbf{Step 1: Shift of the DFE and choice of the parameter.}
Let
\[
x_1 = S_N - S_N^0,\quad
x_2 = S_S - S_S^0,\quad
x_3 = I_N,\quad
x_4 = I_S,\quad
x_5 = R,
\]
and define $x=(x_1,\dots,x_5)^\top$. Then $x=0$ corresponds to the
disease-free equilibrium $E^0$. We take as bifurcation parameter the
transmission coefficient $\beta$ and write the system in the form
$\dot x = f(x,\beta)$, with $f(0,\beta)\equiv 0$ for all $\beta$.

As observed above, the basic reproduction number satisfies
$\mathcal R_0(\beta) = \beta\kappa$. Let $\beta^* = 1/\kappa$ be the
unique value such that $\mathcal R_0(\beta^*)=1$, and introduce the
shifted parameter
\[
\psi = \beta - \beta^*.
\]
Then $\psi=0$ if and only if $\mathcal R_0=1$, and $\psi<0$ (resp.\
$\psi>0$) corresponds to $\mathcal R_0<1$ (resp.\ $\mathcal R_0>1$).
Thus system~\eqref{sys:model-ODE-constant-stress} can be written as
\eqref{Castillo-flower} with $f(0,\psi)\equiv 0$.

\medskip\noindent
\textbf{Step 2: Verification of assumption $A_1$.}
Let $J(E^0)=D_x f(0,0)$ denote the Jacobian matrix of
\eqref{Castillo-flower} at  DFE $E^0$
with $\beta=\beta^*$.  A direct
computation from~\eqref{sys:model-ODE-constant-stress} gives
\[
J(E^0) =
\begin{pmatrix}
-(\alpha+\mu) & 0 & -\beta^* S_N^0 & -\sigma\beta^* S_N^0 & 0\\[2pt]
 \alpha       & -\mu & -\beta^* S_S^0 & -\sigma\beta^* S_S^0 & 0\\[2pt]
 0            & 0 & \beta^* S_N^0 - \nu_N & \sigma\beta^* S_N^0 & 0\\[2pt]
 0            & 0 & \beta^* S_S^0 & \sigma\beta^* S_S^0 - \nu_S & 0\\[2pt]
 0            & 0 & \gamma_N & \gamma_S & -\mu
\end{pmatrix}.\]
The eigenvalues of $J(E^0)$ are 
$- \mathcal R_0(\beta^*)$, $0$, $- \alpha - \mu$ and  $- \mu$ with multiplicity $2$.
Therefore, at $\psi=0$ the Jacobian
$J(E^0)$ has exactly one eigenvalue equal to $0$, while all other
eigenvalues are real and strictly negative. This verifies assumption
$A_1$ of the Theorem~4.1 in \cite{castillo2004dynamical}.

\medskip\noindent
\textbf{Step 3: Verification of assumption $A_2$ and computation of
$a$ and $b$.}
We need to compute the left and
    right eigenvectors of $J(E^0)$ associated with the eigenvalue $0$. 

Let $v=(v_1,v_2,v_3,v_4,v_5)$ be a left eigenvector associated with the
eigenvalue $0$, that is
\[
v\,J(E^0) = (0,0,0,0,0).
\]
Solving the above equation, 
we obtain

\[
v =  \left( 0,\;0,\;\frac{\nu_S}{\nu_N\sigma}v_4,\;v_4,\;0 \right),\; v_4>0.
\]

Let $u=(u_N,u_S,u_3,u_4,u_5)^\top$ be a right eigenvector corresponding
to the eigenvalue $0$, so that
\[
J(E^0)\,u = (0,0,0,0,0)^\top.
\]
Solving the above system  yields
\[
u 
=
\begin{pmatrix}- \frac{S_N^0 \mu \nu_{1} \nu_{2}}{S_N^0 \alpha \gamma_{1} \nu_{2} + S_N^0 \gamma_{1} \mu \nu_{2} + S_S^0 \alpha \gamma_{2} \nu_{1} + S_S^0 \gamma_{2} \mu \nu_{1}}u_5\\- \frac{S_N^0 \alpha \nu_{1} \nu_{2} + S_S^0 \alpha \nu_{1} \nu_{2} + S_S^0 \mu \nu_{1} \nu_{2}}{S_N^0 \alpha \gamma_{1} \nu_{2} + S_N^0 \gamma_{1} \mu \nu_{2} + S_S^0 \alpha \gamma_{2} \nu_{1} + S_S^0 \gamma_{2} \mu \nu_{1}}u_5\\\frac{S_N^0 \mu \nu_{2}}{S_N^0 \gamma_{1} \nu_{2} + S_S^0 \gamma_{2} \nu_{1}}u_5\\\frac{S_S^0 \mu \nu_{1}}{S_N^0 \gamma_{1} \nu_{2} + S_S^0 \gamma_{2} \nu_{1}}\\u_5\end{pmatrix}, \; u_5>0.
\]

We now compute the coefficients $a$ and $b$ defined in
assumption $A_2$ of Theorem~4.1 in \citep{castillo2004dynamical}. For the rest of the analysis, we consider $v_4=1$ and $u_5=1$. Let $f_k$ denote
the $k$th component of the vector field $f$, corresponding to the
equations for $(x_1,x_2,x_3,x_4,x_5)$. Then
\[
a = \sum_{k,i,j=1}^{5} v_k u_i u_j
      \frac{\partial^2 f_k}{\partial x_i \partial x_j}(0,0),
\qquad
b = \sum_{k,i=1}^{5} v_k u_i
      \frac{\partial^2 f_k}{\partial x_i \partial \beta}(0,0).
\]

\emph{Computation of $b$.} Only those terms of $f$ that depend
explicitly on the parameter $\beta$ contribute to $b$. These are the
infection terms involving $\lambda = \beta(x_3+\sigma x_4)$ in the
equations for $x_1,x_2,x_3,x_4$. A computation at
$E^0$ shows that the only non-null second derivatives
are
\[
\begin{aligned}
&\frac{\partial^2 f_N}{\partial x_3\partial\beta}(0,0) = -S_N^0, \quad
 &&\frac{\partial^2 f_N}{\partial x_4\partial\beta}(0,0) = -\sigma S_N^0,\\
&\frac{\partial^2 f_S}{\partial x_3\partial\beta}(0,0) = -S_S^0, \quad
 &&\frac{\partial^2 f_S}{\partial x_4\partial\beta}(0,0) = -\sigma S_S^0,\\
&\frac{\partial^2 f_3}{\partial x_3\partial\beta}(0,0) =  S_N^0, \quad
 &&\frac{\partial^2 f_3}{\partial x_4\partial\beta}(0,0) =  \sigma S_N^0,\\
&\frac{\partial^2 f_4}{\partial x_3\partial\beta}(0,0) =  S_S^0, \quad
 &&\frac{\partial^2 f_4}{\partial x_4\partial\beta}(0,0) =  \sigma S_S^0,
\end{aligned}
\]
 Substituting these into the expression for $b$ and using the
fact that $u_3,u_4,v_3,v_4>0$, and using the fact that 

\[b=v_3\sum_{i=1}^{5}u_i\frac{\partial f_3}{\partial x_i \partial \beta}+ v_4\sum_{i=1}^{5}u_i\frac{\partial f_4}{\partial x_i \partial \beta}=(\frac{ \nu_S}{\nu_N}\sigma S_N^0+S_S^0)(u_3+\sigma u_4) \]
we obtain

\begin{equation}\label{exp:b}
    b= \frac{\mu \left(S_N^0 \nu_{2} + S_S^0 \nu_{1} \sigma\right)^{2}}{\nu_{1} \sigma \left(S_N^0 \gamma_{1} \nu_{2} + S_S^0 \gamma_{2} \nu_{1}\right)} >0.
\end{equation}

\emph{Computation of $a$.} The coefficient $a$ depends on the
second derivatives of $f$ with respect to the state variables $x_i$ and
$x_j$. In system~\eqref{sys:model-ODE-constant-stress} the only nonlinearities are the
bilinear incidence terms $\lambda x_1$ and $\lambda x_2$, so only the
equations for $x_1,x_2,x_3,x_4$ contribute to $a$. 
\[
\begin{aligned}
&\frac{\partial^2 f_N}{\partial x_1 \partial x_3}(0,0) = \frac{\partial^2 f_N}{\partial x_3 \partial x_1}(0,0) = \frac{\partial^2 f_S}{\partial x_2 \partial x_3}(0,0) =\frac{\partial^2 f_S}{\partial x_3 \partial x_2}(0,0) = -\beta^*,\\[4pt]
&\frac{\partial^2 f_N}{\partial x_1 \partial x_4}(0,0) = \frac{\partial^2 f_N}{\partial x_4 \partial x_1}(0,0) =\frac{\partial^2 f_S}{\partial x_2 \partial x_4}(0,0) =\frac{\partial^2 f_S}{\partial x_4 \partial x_2}(0,0) = -\sigma\beta^*,\\[6pt]
&\frac{\partial^2 f_3}{\partial x_1 \partial x_3}(0,0) = \frac{\partial^2 f_3}{\partial x_3 \partial x_1}(0,0) = \frac{\partial^2 f_4}{\partial x_2 \partial x_3}(0,0) = \frac{\partial^2 f_4}{\partial x_3 \partial x_2}(0,0) = \beta^*,\\[4pt]
&\frac{\partial^2 f_3}{\partial x_1 \partial x_4}(0,0) = \frac{\partial^2 f_3}{\partial x_4 \partial x_1}(0,0) = \frac{\partial^2 f_4}{\partial x_2 \partial x_4}(0,0) = \frac{\partial^2 f_4}{\partial x_4 \partial x_2}(0,0) = \sigma\beta^*,
\end{aligned}
\]

Moreover, the terms will be only for non-null vectors $v_3$ and $v_4$ which correspond to

\[
\begin{aligned}
a
&= \sum_{i,j=1}^{5} v_3 u_i u_j
      \frac{\partial^2 f_3}{\partial x_i \partial x_j}(0,0)
  + \sum_{i,j=1}^{5} v_4 u_i u_j
      \frac{\partial^2 f_4}{\partial x_i \partial x_j}(0,0)\\[4pt]
&= v_3\Bigl(
    u_N u_3\,\tfrac{\partial^2 f_3}{\partial x_1 \partial x_3}
  + u_3 u_N\,\tfrac{\partial^2 f_3}{\partial x_3 \partial x_1}
  + u_N u_4\,\tfrac{\partial^2 f_3}{\partial x_1 \partial x_4}
  + u_4 u_N\,\tfrac{\partial^2 f_3}{\partial x_4 \partial x_1}
  \Bigr)\\
&\quad
 + v_4\Bigl(
    u_S u_3\,\tfrac{\partial^2 f_4}{\partial x_2 \partial x_3}
  + u_3 u_S\,\tfrac{\partial^2 f_4}{\partial x_3 \partial x_2}
  + u_S u_4\,\tfrac{\partial^2 f_4}{\partial x_2 \partial x_4}
  + u_4 u_S\,\tfrac{\partial^2 f_4}{\partial x_4 \partial x_2}
  \Bigr)\\[4pt]
&= v_3\Bigl(
    2\beta^* u_N u_3 + 2\sigma\beta^* u_N u_4
  \Bigr)
 + v_4\Bigl(
    2\beta^* u_S u_3 + 2\sigma\beta^* u_S u_4
  \Bigr)\\[4pt]
&= 2\beta^*\Bigl[
    v_3\bigl(u_N u_3 + \sigma u_N u_4\bigr)
  + v_4\bigl(u_S u_3 + \sigma u_S u_4\bigr)
  \Bigr].
\end{aligned}
\]

Replacing the expression of $\beta^*$ in \eqref{eq:beta_star} and the expression of each $u_i$, we obtain 

\begin{equation}\label{exp:a}
a = - \frac{2 \mu \nu_{1} \nu_{2}^{2} \left(S_N^0 \alpha \nu_{1} \sigma + S_N^0 \mu \nu_{2} + S_S^0 \alpha \nu_{1} \sigma + S_S^0 \mu \nu_{1} \sigma\right)}{\sigma \left(\alpha + \mu\right) \left(S_N^0 \gamma_{1} \nu_{2} + S_S^0 \gamma_{2} \nu_{1}\right)^{2}} <0.
\end{equation}

\medskip\noindent
\textbf{Step 4:}
We have verified that assumptions $A_1$ and $A_2$ and 
since $b>0$ and $a<0$, we are in  case~4 of
\cite[Theorem~4.1]{castillo2004dynamical}, then the system
undergoes a \emph{forward} transcritical bifurcation at $\psi=0$:
as $\psi$ (equivalently $\beta=\beta^*$, or $\mathcal R_0=1$) passes from negative
to positive values, the equilibrium at the DFE changes
stability from locally asymptotically stable to unstable, and a unique
positive endemic equilibrium  branch  emerges and is locally asymptotically
stable. Translating back to the original variables, this yields a
unique endemic equilibrium $E^*$ for $\mathcal R_0>1$, which is locally
asymptotically stable, while $E^0$ is the only equilibrium and is
locally asymptotically stable for $\mathcal R_0<1$.

This completes the proof of Proposition~\ref{prop:behaviour-constant-stress}.
\end{proof}

\bibliographystyle{plain}
\bibliography{references}

\end{document}